\documentclass[12pt]{article}
\usepackage[T1]{fontenc}
\usepackage{lmodern}
\usepackage[margin=1in]{geometry}
\usepackage{amsmath,amssymb,amsthm,mathtools}
\usepackage{booktabs,array,microtype,needspace}
\usepackage[round,authoryear]{natbib}
\setcitestyle{aysep={,},yysep={;}}
\usepackage[colorlinks=true,allcolors=blue]{hyperref}
\hypersetup{pdftitle={Who Enters the Record? Paywalls, Contributor Selection, and Social Learning},pdfauthor={Maria Kovalenko and Georgy Lukyanov}}
\newtheorem{theorem}{Theorem}
\newtheorem{proposition}{Proposition}
\newtheorem{lemma}{Lemma}

\theoremstyle{remark}
\newcommand{\E}{\mathbb E}
\newcommand{\Pp}{\mathbb P}
\newcommand{\ind}{\mathbf 1}
\newcommand{\sig}{\sigma}

\newcommand{\HH}{\mathcal H}
\title{Who Enters the Record?\\Paywalls, Contributor Selection, and Social Learning}
\author{Maria Kovalenko\thanks{International College of Economics and Finance, National Research University Higher School of Economics, Moscow, Russia. Email: \href{mailto:mvkovalenko@edu.hse.ru}{mvkovalenko@edu.hse.ru}.}\and Georgy Lukyanov\thanks{Toulouse School of Economics, Toulouse, France. Corresponding author. Email: \href{mailto:georgy.lukyanov@tse-fr.eu}{georgy.lukyanov@tse-fr.eu}.}}
\date{September 2026}
\begin{document}
\maketitle
\begin{abstract}
This paper studies social learning when buying access to previous decisions also determines whose decision enters the record. Individuals first receive private information and then decide whether to pay for access. A positive price selects relatively weak private beliefs. With one founding observation, every buyer copies it, and the record acquires no further information. With a richer initial history, buyers can disagree and the record can improve; nevertheless, a price bounded away from zero prevents complete learning, even with unbounded private beliefs. We then consider a rule under which a buyer makes a preliminary choice before access and a revised choice afterward. A positive probability of implementing the preliminary choice makes it follow the private signal. Retaining both choices restores complete learning under a fixed effective fee or pooled myopic pricing, including with bounded beliefs. The rule reduces willingness to pay and sacrifices some current decision accuracy. A worked example shows how the subsequent improvement in information can outweigh both initial losses. The results show why preserving a judgment made before exposure can change what later users learn from the same selected population.
\end{abstract}

\section{Introduction}\label{sec:intro}

When individuals are uncertain about a decision, they often look at the decisions of others. A person considering an investment may consult earlier recommendations; someone evaluating a new product may seek the experience of previous users; and a participant in a forecasting discussion may revise a judgment after reading the views already expressed. In each case, the value of observing others depends on how much the individual already knows. Someone with a strong private reason to choose one alternative has less to gain from consulting the group than someone who is still undecided.

Suppose that obtaining this information requires a payment. Those who rely on their own information remain outside, while those who expect to benefit from the history are willing to buy access. This familiar selection becomes consequential for social learning when the service records only the decisions made by its customers. The person who declines to buy still makes a decision, but that decision is absent from the evidence available to the next customer. The history is therefore produced by a population selected for its willingness to consult the history itself.

To fix ideas, consider a service that sells access to earlier recommendations about a particular project. Each arriving user has conducted some private research and must decide whether the project is worthwhile. A buyer reads the recommendations and records a decision through the service; a nonbuyer decides independently, and the service does not observe that choice. Our question is whether an increasingly long record of customer decisions necessarily provides increasingly reliable guidance, and what the service would have to record to preserve the information of its users. We use this example as a stylized decision environment. The analysis concerns a particular issue, with private research preceding the access decision, rather than a membership purchased before a user knows which issues will arise.\footnote{The payoff below is the user's benefit from a correct decision, normalized to one. It is not a prediction prize paid out of the provider's access revenue. In a laboratory implementation, accuracy can instead be rewarded by the experimenter, with that payment budget kept separate from fee revenue.}

This paper develops a binary-state model in which observation and contribution rights are sold together. Individuals first receive conditionally independent private signals. They then decide whether to pay for the archive and choose the action that is best given their information. The contents of the archive are hidden before payment, but its information structure and recording rule are known. Only a buyer's action enters the record. All individuals understand this selection and update by Bayes' rule.

The one-founder case makes the mechanism particularly transparent. If the initial recommendation has known accuracy, the only individuals who find it worth purchasing are those whose private information is sufficiently weak. At every positive price, their signals are too weak to reverse the recommendation after access. Each buyer therefore adds another copy, while the agents with sufficiently strong signals to contradict it remain outside. Moreover, entry selects on the strength of a signal rather than its direction, so the decision to enter supplies no additional information about the state. The archive remains as informative as its first observation, however many customers it attracts (Proposition~\ref{prop:founder}).

With several founding observations, this immediate copying result needs qualification. Before paying, a buyer may expect to encounter either a strong or a weak history. Information can be worth purchasing because the strong realization would change the buyer's decision, even though the weak realization does not. A buyer may then contradict the recommendation actually observed. We give a three-founder example and show how to update the record correctly. The broader impediment to learning nevertheless remains: a price bounded away from zero puts a uniform bound on the private evidence of every buyer. Recording only buyers' subsequent actions then puts a finite bound on the archive's log odds as well (Theorem~\ref{thm:nonlearning}). The record can improve without eventually revealing the state.

We next ask whether the same selected buyers can supply the missing information. They can, provided that an independent judgment is preserved before exposure. A buyer who subsequently follows the archive may initially have preferred either action, and that initial direction remains informative about the state. Simply asking for a preliminary opinion does not establish an incentive to report it accurately. We therefore consider a consequential choice: before seeing the archive, the buyer commits to a preliminary action; after access, the buyer submits a revised action; and an announced lottery determines which action is implemented. Both choices are retained, with their order clearly identified.

This rule has a direct effect on the access decision. If the preliminary choice is implemented with probability $\varepsilon$, the value of consulting the archive is multiplied by $1-\varepsilon$. The provider can preserve the entry set at a given information experiment by reducing the fee in the same proportion. Conditional on joining, each individual strictly prefers a preliminary choice that follows the private signal. Under either a fixed effective fee or the pooled myopic pricing rule specified below, entry remains sufficiently broad for these choices to be uniformly informative. The resulting archive learns the state even when private beliefs are bounded (Theorem~\ref{thm:learning}). Thus the selected sample is capable of supplying information that its post-access actions alone would suppress.

There is a cost to making the preliminary choice consequential. Sometimes it is implemented even though the revised choice would be better informed. A fixed positive implementation probability therefore leaves a positive limiting decision error, despite complete learning by the archive. This distinction matters for the economic interpretation. In our bounded-signal example, the reporting rule initially lowers both access revenue and decision accuracy. As the record improves, it becomes useful to more individuals, and the later gains outweigh both losses. The example compares outcomes under two specified policies; the value of adopting either policy over the provider's lifetime also depends on the horizon and discounting.

\subsection{Relation to the literature}

The starting point is the literature on informational cascades and observational learning. \citet{bhw1992} show how an individual may rationally disregard private information after observing earlier actions. \citet{smith2000} establish the importance of the distinction between bounded and unbounded private beliefs; \citet{bikhchandani2024} review this literature. In our benchmark, the private signal distribution can be unbounded, but paying for observation selects a bounded part of that distribution. The individuals who could provide sufficiently strong contrary evidence do not become observable contributors.

Costly observation and the decision to acquire social information are studied by \citet{kultti2007} and \citet{song2016}. Song is the closest benchmark for the selection at the access decision: strong private information reduces an individual's incentive to observe predecessors. In his symmetric equilibrium result, sufficiently strong private beliefs and expanding observation capacity allow the observed sequence to reveal the state because nonobservers' actions remain available as evidence. Our change is that the action of a nonobserver is also excluded from the record. This connects the cost of consuming social information to the production of the next observation. \citet{arieli2022} also study how pricing affects social learning. Here the object being purchased is the history itself, and the access decision governs contribution to that history. The reporting rule recovers independent evidence from people who still choose access after receiving their signals.

The coverage of the record is itself a substantive part of a learning environment. \citet{guarino2011} and \citet{herrera2013} examine learning when only one of the available actions is observable. In Herrera and H\"orner, the absence of an investment itself conveys information; here the symmetric access decision is state-neutral. \citet{parakhonyak2023} show how endogenously chosen capacity restrictions can limit learning even with unbounded private signals. Our restriction operates through a uniform fee for observing the history and affects both action directions symmetrically. The comparison with occasional recorded nonbuyers in Section~\ref{sec:outsiders} makes the coverage margin explicit.

Selection also matters when consumers learn from reviews. \citet{ifrach2019} and \citet{acemoglu2022} study the information generated by a selected population of purchasers. The observation here has a different origin: it is a decision made after reading the inherited record, rather than a new experience signal obtained from consuming a product. Access can therefore change not only who contributes, but also what a contribution reveals. \citet{smirnov2025} endogenize communication by purchasers who care about future consumers. Their reviews describe new consumption experiences, and altruism creates strategic distortion. Here users care about their own decisions, and the reporting rule preserves information held before access. More broadly, \citet{markovich2024} study the commercial, private and public benefits of platform data. Our focus is the information that survives when consuming the existing data and contributing to it are governed by the same rule.

Preserving independent judgments before social influence is an established principle. \citet{frey2021} distinguish an independent start from sequential choices exposed to earlier decisions. \citet{peng2025} show how selective disclosure of signal-consistent actions can preserve private information, although spontaneous disclosure in their experiment is not selective enough to improve learning. Their disclosure choice has no direct monetary payoff consequence. We instead make the preliminary choice consequential and study how the rule affects both paid entry and the information supplied by those who enter. Random selection among consequential decisions has an established incentive rationale \citep{azrieli2018}; we verify the incentives directly for our timing and payoff assumptions. The analysis then connects the resulting willingness to pay to continued participation and complete learning within the selected population. Advance purchase provides a separate contract-timing comparison with familiar pricing logic \citep{xie2001,nocke2011}.

The rest of the paper proceeds as follows. Section~\ref{sec:model} presents the environment and the closed-record benchmark. Section~\ref{sec:preliminary} studies preliminary choices, entry, learning and the cost of implementation. Section~\ref{sec:richer} considers richer initial histories and occasional observation of nonbuyers. Section~\ref{sec:conclusion} concludes. Proofs and additional calculations appear in the appendix.

\section{The environment and a closed-record benchmark}\label{sec:model}

\subsection{Private information and the record}

There is an unknown state $\theta\in\{-1,+1\}$, with probability one half on each state. Individuals arrive sequentially. Each observes a private log-likelihood ratio $X$ before deciding whether to purchase social information. Signals are conditionally independent across individuals and independent of the initial record given the state. We work directly with densities satisfying
\begin{equation}
 f_+(x)=e^x f_-(x),\qquad f_+(x)=f_-(-x).
 \label{eq:signals}
\end{equation}
The densities are continuous and positive on $(-\bar x,\bar x)$, where $0<\bar x\leq\infty$. When $\bar x$ is finite, private beliefs are bounded; otherwise they are unbounded. Define $\sig(x)=(1+e^{-x})^{-1}$. An individual with private log odds $x$ alone chooses $\operatorname{sign}x$ and achieves expected accuracy $\sig(|x|)$.

An individual obtains payoff $\ind\{a=\theta\}$ from the implemented binary action $a$, less any access fee. The individual has no payoff from influencing later users and no additional preference over what the record says. The fee is a transfer to the provider. There are no further costs of reading in the model.

In the benchmark, the initial record consists of one founding recommendation $Y\in\{-1,+1\}$. Its state-conditional accuracy and log-likelihood strength are
\begin{equation}
 \Pp(Y=\theta\mid\theta)=r\in(1/2,1),\qquad
 h=\log\frac{r}{1-r}.
 \label{eq:founder}
\end{equation}
The founder's recommendation is based on information independent of subsequent users' signals conditional on the state. We shall also allow a symmetric initial experiment $\HH_0$ with finite log odds $L_0$ almost surely. Symmetry means that reflecting all directions in a realization swaps its two state-conditional probabilities. A finite collection of independent founding recommendations is one example. The one-founder case has $L_0=hY$.

Let $\HH_n$ denote the full archive after $n$ completed member records, and let
\begin{equation}
 L_n=\log\frac{\Pp(\theta=+1\mid\HH_n)}{\Pp(\theta=-1\mid\HH_n)}.
 \label{eq:logodds}
\end{equation}
The archive preserves the initial observations, subsequent contributions and their order. Its contents are hidden before payment. An arriving individual knows its size, the recording rule, the signal technology and the posted fee, but sees neither its direction nor its realized confidence. A buyer obtains the full archive. Its log odds summarize all the information relevant to the binary decision.\footnote{A displayed count that suppresses which observations preceded exposure would be a different information product. Here readers know the origin of the data and interpret dependent observations correctly. In particular, repeated agreement does not cause Bayesian readers to become overconfident.}

The posted fee at size $n$ is the same across all private archive realizations of that size. We refer to this as a \emph{pooled} fee. A prescribed size-dependent rule, or a pricing procedure that does not inspect content, has this form. Later we allow the fee to maximize current expected revenue within this class. We do not allow a privately informed provider to communicate through its choice of price.

\subsection{Entry and the value of observation}\label{sec:access}

Under the initial recording rule, an individual observes $X$ and chooses whether to pay $c_n$. A buyer reads $\HH_n$, chooses an optimal action, and that action is appended. A nonbuyer chooses from private information alone, and the action is absent from the archive. Thus $n$ counts customers rather than all arrivals.

Write $\nu_{\theta,n}$ for the state-conditional distribution of the archive log odds. Given $X=x$, the value of reading the archive is the increase in optimal expected accuracy over $\sig(|x|)$. Denote this increase by $V_n(x)$. Also put
\begin{equation}
 T=|X|,\qquad D(b)=\Pp(T\leq b).
 \label{eq:demand}
\end{equation}
The distribution of $T$ is identical in the two states by symmetry. If the signal support is bounded, $D(b)=1$ for $b\geq\bar x$.

\Needspace{9\baselineskip}
\begin{lemma}\label{lem:access}
For a symmetric archive experiment and $x\geq0$,
\begin{equation}
 V_n(x)=\frac{1}{1+e^x}\int_{l<-x}(1-e^{x+l})\,\nu_{-,n}(dl).
 \label{eq:generalvalue}
\end{equation}
The value is even, continuous and strictly decreasing where positive, and
\begin{equation}
 V_n(x)\leq1-\sig(|x|).
 \label{eq:perfectbound}
\end{equation}
For $0<c_n<V_n(0)$, entry occurs on an interval $|X|\leq b_n$, up to null ties, with the cutoff capped at $\bar x$ if all types buy. Entry and nonentry are uninformative about the state, including conditional on the actual archive.
\end{lemma}

The formula has a simple interpretation. Someone with $x\geq0$ would choose the positive action without access. Information is useful only when the archive is sufficiently negative to reverse that choice. The integral in \eqref{eq:generalvalue} adds the improvement over those realizations. As private confidence increases, reversal requires stronger contrary information, and the value of observation falls.

The neutrality of entry follows from the same symmetry. Buying reveals that $|X|$ is small enough, but gives no indication of the sign of $X$. Since an incoming signal is conditionally independent of the existing record, this remains true for someone who already knows the archive's contents. Waiting times and the number of nonbuyers therefore provide no additional directional news under the stated rules.

\subsection{Why a one-founder record stops learning}

With the founding recommendation alone, the value function is especially simple:
\begin{equation}
 V_0(x)=[r-\sig(|x|)]_+.
 \label{eq:foundervalue}
\end{equation}
When $|x|\geq h$, the recommendation never changes the individual's action. When $|x|<h$, following it is optimal in either direction and gives accuracy $r$. The value in \eqref{eq:foundervalue} is the difference between this accuracy and the private alternative.

For a fixed fee $0<c<r-1/2$, define the uncapped cutoff
\begin{equation}
 \beta(c)=\log\frac{r-c}{1-r+c}<h,
 \qquad b(c)=\min\{\bar x,\beta(c)\}.
 \label{eq:foundercutoff}
\end{equation}
The strict inequality is the source of the copying result.

\Needspace{7\baselineskip}
\begin{proposition}\label{prop:founder}
Suppose that only buyers' post-access actions are recorded, the initial archive is $Y$, and $0<c<r-1/2$. A buyer enters if and only if $|X|\leq b(c)$ and then chooses $Y$. After every number of member actions,
\begin{equation}
 L_n=hY,\qquad \Pp(\theta=Y\mid\HH_n)=r.
 \label{eq:freeze}
\end{equation}
There are infinitely many buyers almost surely. Their actions are all wrong with probability $1-r$.
\end{proposition}

Neither a preference for conformity nor a misunderstanding of the record is needed. The buyers correctly see that earlier customer decisions are copies. They buy because the founding information is still useful to them. Conversely, an individual whose signal could reverse the founder's recommendation has no reason to pay for it and contributes nothing. The access decision removes precisely the observations that could break the agreement.

Any positive fee with positive demand produces the same result, including a current-revenue-maximizing fee in this benchmark. That conclusion relies on the initial recommendation having a known strength. With a richer archive, a buyer can find the expected experiment worth purchasing yet encounter a weak realization. We return to this distinction in Section~\ref{sec:richer}. First, we examine a way to preserve the information of the buyers themselves.

\section{Retaining an independent judgment}\label{sec:preliminary}

\subsection{A consequential preliminary choice}

A person who ends up following the archive need not have agreed with it beforehand. That preliminary preference is valuable to successors because it reflects information acquired independently of the history. Requiring a report before access, however, is insufficient if the report has no consequence for the person making it.

We therefore modify the decision rule as follows. At size $n$, the provider announces a fee $c_n$ and a probability $0<\varepsilon_n\leq\bar\varepsilon<1$. The probability schedule is a prescribed function of public size, so its announcement reveals no archive content. After observing $X$, a buyer pays and commits to the rule. Before access the buyer submits an irrevocable preliminary choice $B$. The buyer then reads a frozen snapshot of $\HH_n$ and submits a revised choice $A$. An independent lottery implements $B$ with probability $\varepsilon_n$ and $A$ otherwise. Both choices enter the archive, identified as preliminary and revised, after the second choice has been fixed.

The snapshot, fee and lottery probability do not depend on $B$. The buyer cannot withdraw selectively after access or replace the action selected by the lottery. These restrictions ensure that the preliminary choice changes only the chance of making that decision. It cannot improve the information subsequently purchased or alter the buyer's other payoffs. The mechanism requires enforceable choices; it does not apply directly to a freely retractable comment.

All observations are collected before the state is resolved or action payoffs are revealed. More precisely, consider finite economies with $N$ completed member records, followed by the implementation lotteries and settlement. There are no delay costs, and the decisions of a fixed initial sequence of members do not depend on $N$. The learning limit is the limit of these consistent pre-settlement records as $N$ grows. Public outcome revelation is therefore not the source of the learning result.\footnote{This timing is natural for decisions collected before an event or a common decision deadline. Immediate feedback about the same state would change the information available to subsequent individuals. The mechanism does not require the provider to know the state when it collects the choices or selects which one to implement.}

\Needspace{10\baselineskip}
\begin{proposition}\label{prop:elicitation}
Under the preliminary-choice rule, optimal choices are, almost surely,
\begin{equation}
 B=\operatorname{sign}X,\qquad A=\operatorname{sign}(X+L_n).
 \label{eq:twodecisions}
\end{equation}
The preliminary choice is strict whenever $X\ne0$. The gross value of access is $(1-\varepsilon_n)V_n(x)$. Consequently, if
\begin{equation}
 p_n=\frac{c_n}{1-\varepsilon_n}
 \label{eq:effectiveprice}
\end{equation}
is the effective fee, entry occurs when $V_n(X)\geq p_n$. It is a symmetric cutoff decision and remains uninformative about the state.
\end{proposition}

To see the incentive, suppose $x>0$. Reporting the positive preliminary choice rather than the negative one raises expected payoff by
\begin{equation}
 \varepsilon_n[2\sig(x)-1]=\varepsilon_n\tanh(x/2)>0.
 \label{eq:incentivegap}
\end{equation}
The same argument, with signs reversed, applies to $x<0$. At $\varepsilon_n=0$, the preliminary message has no payoff consequence: truthful and uninformative reporting are both compatible with individual optimization. Positive implementation probability makes its direction strict. The revised decision is optimal after combining the private signal with the archive. Before access, the buyer's expected gross payoff is therefore
\begin{equation}
 \varepsilon_n\sig(|x|)+(1-\varepsilon_n)[\sig(|x|)+V_n(x)]
 =\sig(|x|)+(1-\varepsilon_n)V_n(x).
 \label{eq:payoff}
\end{equation}
The reduction in willingness to pay follows by comparing this with the nonbuyer's payoff $\sig(|x|)$.

The rule elicits a binary direction rather than a full posterior belief. This is enough for the learning result. Define the accuracy of the preliminary choice among buyers at cutoff $b$ by
\begin{equation}
 q_b=\Pp_+(X>0\mid |X|\leq b)
     =\E[\sig(T)\mid T\leq b]>1/2.
 \label{eq:selectedaccuracy}
\end{equation}
The equality follows from \eqref{eq:signals}: conditional on $T=t$ and either state, the sign of the signal is correct with probability $\sig(t)$. Selection retains individuals with weaker information, but their directions still favor the true state. Also, $q_b$ is nondecreasing in $b$.

\subsection{Entry as the record improves}\label{sec:pricing}

It is not enough that every preliminary choice be informative at the date it is recorded. If the cutoff became arbitrarily small, successive reports could become arbitrarily close to fair coin tosses. We show that this problem does not arise under two simple pricing policies.

The first policy keeps the effective fee fixed at $p\in(0,V_0(0))$, so the actual fee is $c_n=(1-\varepsilon_n)p$. Because the archive retains its earlier observations, a later buyer can always ignore the additional data. Thus the value of access cannot decrease:
\begin{equation}
 V_{n+1}(x)\geq V_n(x)\geq V_0(x).
 \label{eq:refinement}
\end{equation}
It follows that $b_{n+1}\geq b_n\geq b_0>0$. This comparison is between the experiments offered at successive sizes. Realized confidence can fall when a particular new observation contradicts the existing record.

The second policy chooses the pooled fee to maximize current expected revenue per arriving individual. Define
\begin{equation}
 R_n(p)=p\Pp(V_n(X)\geq p),\quad
 R_n^*=\max_{p\in[0,1/2]}R_n(p),\quad
 p_n^*\in\arg\max_{p\in[0,1/2]}R_n(p).
 \label{eq:revenue}
\end{equation}
Select a maximizer by a fixed deterministic rule. Proposition~\ref{prop:elicitation} implies that actual revenue is $(1-\varepsilon_n)R_n(p)$, so the provider posts $c_n^*=(1-\varepsilon_n)p_n^*$. Assume $V_0(0)>0$, ensuring positive initial revenue. Full retention gives $R_n^*\geq R_0^*>0$. Since demand is at most one and the effective fee is at most $1/2$,
\begin{equation}
 p_n^*\geq R_0^*,\qquad D(b_n^*)\geq2R_0^*>0.
 \label{eq:positivebounds}
\end{equation}
The provider cannot maximize a revenue bounded away from zero while admitting a vanishing share of buyers at a bounded price. Hence the cutoff is bounded away from zero under this policy as well.

The provider in this comparison maximizes current revenue using the distribution of archives at the publicly known size. The pooled restriction prevents price from revealing private archive contents. Allowing such signaling, or accounting for the effect of today's price on future profits, would require a different pricing problem.\footnote{The same distinction applies to \eqref{eq:refinement}. An archive that physically retains earlier data is a more informative experiment for a buyer with access. It does not follow that an arbitrary public statistic of that archive, or a privately chosen posted price, preserves the same pre-purchase information problem.}

\subsection{Learning from the selected buyers}

\begin{theorem}\label{thm:learning}
Suppose that the symmetric initial archive has finite log odds almost surely and $V_0(0)>0$. Under the preliminary-choice rule, full retention, and either pricing policy in Section~\ref{sec:pricing}, there is a deterministic $b_{\min}>0$ such that $b_n\geq b_{\min}$ at every size. There are infinitely many members, and
\begin{equation}
 \Pp(\theta=+1\mid\HH_n)\longrightarrow\ind\{\theta=+1\}
 \quad\text{almost surely}.
 \label{eq:learning}
\end{equation}
If $q_{\min}=q_{b_{\min}}$, the ex ante Bayes error of the archive after $n$ member records is at most
\begin{equation}
 \exp\{-2n(q_{\min}-1/2)^2\}.
 \label{eq:errorbound}
\end{equation}
Members' revised-choice error converges to zero. These conclusions hold with bounded as well as unbounded private beliefs.
\end{theorem}

The central observation is that the preliminary choices remain independent conditional on the state. Entry depends on signal strength, but its cutoff is fixed by the publicly known size. Consequently, the next selected sign is a new draw whose accuracy is at least $q_{\min}>1/2$. Even a reader who ignores every revised choice and simply counts the preliminary directions can eventually identify the state. The full Bayesian record performs at least as well. The proof in Appendix~\ref{app:learning} supplies the probability bound and the almost-sure statement.

The two choices of one person, however, are generated by the same signal and must be interpreted jointly. For a realized archive log odds $l$ and cutoff $b$, a pair $(u,v)$ identifies the signal interval
\begin{equation}
 C_{u,v}(l,b)=[-b,b]\cap\{x:\operatorname{sign}x=u\}
                         \cap\{x:\operatorname{sign}(x+l)=v\}.
 \label{eq:paircell}
\end{equation}
The record's likelihood increment is
\begin{equation}
 \log\frac{\int_{C_{u,v}(l,b)}f_+(x)\,dx}
                {\int_{C_{u,v}(l,b)}f_-(x)\,dx},
 \label{eq:pairupdate}
\end{equation}
on a positive-probability branch, with integrals restricted to the signal support. This joint likelihood is what successors use. When $|l|\geq b$, the revised choice is a copy, but the preliminary choice still contributes the nonzero increment $u\log[q_b/(1-q_b)]$.

Theorem~\ref{thm:learning} therefore makes a distinction within the same population. Buyers are selected for weak private confidence under both recording rules. Under the first rule, this selection eventually prevents their actions from correcting the record. Under the second, their preliminary directions keep adding evidence. Expanding contributor eligibility to people who do not buy is one possible remedy, considered below, but it is not necessary for this result.

\subsection{The cost of making a judgment consequential}\label{sec:cost}

Let $e_{A,n}$ be the error probability of the revised choice among members at size $n$, and let $e_{{\rm imp},n}$ be the error of the choice actually implemented. Write $q_{\rm all}=\E[\sig(T)]<1$ for private-signal accuracy in the whole population. Then
\begin{equation}
 e_{{\rm imp},n}=\varepsilon_n(1-q_{b_n})+(1-\varepsilon_n)e_{A,n}.
 \label{eq:implementederror}
\end{equation}
The first term persists if the preliminary choice keeps a fixed positive chance of being selected.

\Needspace{10\baselineskip}
\begin{proposition}\label{prop:cost}
Under the conditions of Theorem~\ref{thm:learning}, a fixed $\varepsilon>0$ gives
\begin{equation}
 \liminf_n e_{{\rm imp},n}\geq\varepsilon(1-q_{\rm all})>0.
 \label{eq:errorfloor}
\end{equation}
If $\varepsilon_n\to0$ while remaining strictly positive, member implemented-choice error converges to zero. At size $n$, the expected accuracy sacrificed by the lottery, relative to implementing the revised choice in the same information environment, is
\begin{equation}
 C_n=\varepsilon_n\E[V_n(X)\mid\text{entry}]\leq\varepsilon_n/2
 \label{eq:cost}
\end{equation}
per member. At a given archive experiment, optimal current fee revenue is reduced by the factor $1-\varepsilon_n$ relative to wholly revised decisions.
\end{proposition}

The opportunity cost in \eqref{eq:cost} is distinct from the fee transfer. No additional elicitation payment is part of the provider's budget, although operating an enforceable rule may itself be costly outside the model. The comparison holds the information environment fixed: it does not assume that deleting preliminary reports from the equilibrium would leave subsequent choices unchanged.

For fixed effective prices, changing the strictly positive lottery probabilities does not change either optimal reporting or the entry set at any given archive experiment. By induction, the distribution of records is the same for any such probability schedule when actual fees are scaled accordingly. This explains why vanishing stakes can remove the limiting implementation loss without stopping learning in the ideal model. Their incentive interpretation needs care, however. Equation~\eqref{eq:incentivegap} becomes arbitrarily small as $\varepsilon_n$ decreases, and any additional motive to misreport can then matter.\footnote{A strictly positive summable schedule makes the ideal limit especially clear. If $\varepsilon_n=\eta 2^{-n-1}$ for $0<\eta<1$, total expected opportunity cost across members is at most $\eta/2$, and only finitely many preliminary choices are selected almost surely. Each remains strictly optimal when reported. This observation is a consequence of exact expected-payoff maximization, not a claim of robustness to reporting noise or expressive preferences.}

Also, correct revised decisions among members do not imply correct decisions throughout the population. Nonbuyers still act on their private information. With unbounded signals and a fixed positive effective fee, some continue to remain outside even when the archive becomes very accurate.

\subsection{A worked comparison}\label{sec:example}

Consider bounded signals on $[-1,1]$ with densities
\begin{equation}
 f_+(x)=\frac{e^{x/2}}{4\sinh(1/2)},\qquad f_-(x)=f_+(-x).
 \label{eq:exampledensity}
\end{equation}
These satisfy \eqref{eq:signals}. The entry probability and selected preliminary accuracy are
\begin{equation}
 D(b)=\frac{\sinh(b/2)}{\sinh(1/2)},\qquad q_b=\sig(b/2).
 \label{eq:exampleclosed}
\end{equation}
Let the founder's accuracy be $r=3/4$. Compare the post-access-only rule at fee $p=0.10$ with the preliminary-choice rule at effective fee $p$, implementation probability $\varepsilon=0.05$, and actual fee $c=0.095$.

Both rules begin with cutoff $b_0=\log(13/7)$, entry probability $D(b_0)=0.603509$, and preliminary accuracy $q_{b_0}=0.576768$. Under the first rule all buyers copy the founder, so the archive and outcomes remain fixed. Under the second rule some initially implement their less accurate preliminary choices and pay a smaller fee. The immediate comparison is therefore unfavorable to the new rule.

As preliminary observations accumulate, the archive improves and the cutoff expands. In the limit of a perfectly informative archive, even the strongest private type values access at $1-\sig(1)=0.268941>p$. Thus all types eventually buy under the second rule. Limiting fee revenue is $0.095$ per arriving individual, and implemented accuracy converges to
\begin{equation}
 1-\varepsilon[1-\sig(1/2)]=0.981123.
 \label{eq:examplelimit}
\end{equation}

\begin{table}[htbp]
\centering\small
\caption{Expected outcomes per arriving individual}
\begin{tabular}{lrrr}
\toprule
Recording rule and archive size & Entry & Fee revenue & Decision accuracy\\
\midrule
Post-access only, every size & 0.603509 & 0.060351 & 0.727006\\
Preliminary and revised, $n=0$ & 0.603509 & 0.057333 & 0.721779\\
Preliminary and revised, $n=4$ & 0.715661 & 0.067988 & 0.738296\\
Preliminary and revised, $n=12$ & 1.000000 & 0.095000 & 0.819566\\
Preliminary and revised, $n\to\infty$ & 1.000000 & 0.095000 & 0.981123\\
\bottomrule
\end{tabular}
\label{tab:comparison}
\end{table}

Table~\ref{tab:comparison} separates finite records from the analytical limit. The finite calculations enumerate every possible history using the joint update \eqref{eq:pairupdate}; the limiting row follows from Theorem~\ref{thm:learning}. Decision accuracy includes nonbuyers and refers to the action implemented, not just the revised recommendation. Appendix~\ref{app:example} gives the calculation and a sufficient bound on $\varepsilon$ for the long-run gains in revenue and accuracy to hold simultaneously. The inequalities in this example are strict, so the comparison persists under small parameter changes.

The example identifies an economic reason to preserve preliminary choices. Whether the provider would adopt the rule depends on how it values the later revenue gains against the initial losses. A short horizon or strong discounting could make the initial sacrifice unattractive even though the record eventually becomes more informative.

\section{Richer histories and the coverage of the record}\label{sec:richer}

The one-founder benchmark combines two properties: the record fails to reveal the state, and every new buyer immediately copies the founding recommendation. These properties need not coincide. We first show why the information disclosed before purchase matters for copying. We then establish the broader nonlearning result and examine the consequence of occasionally observing people outside the paying group.

\subsection{What buyers know about archive strength}\label{sec:confidence}

Suppose the symmetric founding experiment has log odds $L_0$, and its magnitude $H=|L_0|$ is disclosed before purchase while its direction remains hidden. The magnitude is state-neutral. Conditional on $H=h>0$, the remaining experiment has accuracy $\sig(h)$, so
\begin{equation}
 V_h(x)=[\sig(h)-\sig(|x|)]_+.
 \label{eq:disclosedvalue}
\end{equation}
At every positive fee with positive demand, all admitted signals satisfy $|X|<h$. Thus, when only post-access actions are retained, every added action agrees with the initial Bayesian recommendation and $L_n=L_0$ for all $n$. This reasoning applies to several founders as well as one. If the founding vote is tied, $h=0$ and there is no positive-price demand. The original founding entries themselves need not be unanimous.

Now keep the realized confidence hidden until access, as in the main model. Take three independent founding recommendations, each correct with probability $3/4$, and put $h=\log3$. The possible initial log odds are $-3h,-h,h,3h$, with conditional probabilities
\begin{equation}
 64\nu_{+,0}=(1,9,27,27),\qquad
 64\nu_{-,0}=(27,27,9,1)
 \label{eq:threefounders}
\end{equation}
in that order. Consider an individual with $x=2h=\log9$. The private decision is correct with probability $9/10$. Reading the archive changes that decision only if $L_0=-3h$, giving access value
\begin{equation}
 V_0(2h)=\frac1{10}\frac{27}{64}-\frac9{10}\frac1{64}
        =\frac9{320}=0.028125.
 \label{eq:counterexample}
\end{equation}
At fee $c=0.02$, the individual strictly buys. But if the realized archive is $L_0=-h$, the combined log odds are $h>0$: the buyer chooses the positive action, contradicting the negative recommendation. Both inequalities are strict and hold for an interval of signals, so this occurs with positive probability whenever the private signal support includes that interval.

The buyer purchases an experiment before knowing which realization will be supplied. Its strong realizations make it worth buying, even though a weak realization can be outweighed by private information. Substituting the realized confidence into the one-founder cutoff before the purchase decision would miss this effect.

\subsection{A positive price bounds what final choices reveal}\label{sec:nonlearning}

For a realized archive log odds $l$ and buyer cutoff $b$, a member chooses the positive action if $X\geq-l$. Conditional on entry, its probability in state $\theta$ is
\begin{equation}
 Q_\theta(l,b)=\frac{\int_{-l}^{b} f_\theta(x)\,dx}{D(b)}
 \quad\text{if }|l|<b,
 \label{eq:actionprob}
\end{equation}
with $Q_\theta=1$ if $l\geq b$ and $Q_\theta=0$ if $l\leq-b$. A positive action adds $\log(Q_+/Q_-)$ to the archive's log odds; a negative action adds $\log[(1-Q_+)/(1-Q_-)]$. Within $|l|<b$, both actions are informative because $Q_+>Q_-$. Outside that region, all admitted signals imply the same action.

The value bound in Lemma~\ref{lem:access} now has a useful implication. If fees satisfy $c_n\geq\underline c>0$, every buyer must have
\begin{equation}
 |X|\leq \bar b=\log\frac{1-\underline c}{\underline c},
 \label{eq:cap}
\end{equation}
where $\underline c<1/2$; larger fees produce no demand. Even perfect information would not justify paying the fee for a stronger private signal.

\Needspace{10\baselineskip}
\begin{theorem}\label{thm:nonlearning}
Suppose only buyers' post-access actions are retained, the symmetric initial record has finite $L_0$ almost surely, and pooled fees satisfy $c_n\geq\underline c>0$. Under the maintained information policy,
\begin{equation}
 |L_n|\leq\max\{|L_0|,2\bar b\}\quad\text{for every }n,
 \label{eq:globalcap}
\end{equation}
with $\bar b$ given by \eqref{eq:cap}. The archive posterior has an interior limit and does not completely reveal the state. If $|L_0|\leq K_0$ surely and $K=\max\{K_0,2\bar b\}$, archive-only Bayes error is at least $\sig(-K)>0$ at every size and in the limit.
\end{theorem}

The argument is short. Each action is generated by a subset of admitted signals, so its likelihood ratio is an average of $e^x$ over a subset of $[-\bar b,\bar b]$. Its log-likelihood increment is therefore no larger than $\bar b$ in absolute value. Once $|L_n|\geq\bar b$, all possible buyers copy and no further information is added. Before reaching that region, a single update can take the archive at most to $2\bar b$. This proves the bound without requiring immediate copying, a constant realized precision or a constant purchase cutoff.

The result remains relevant under endogenous fees. For the pooled myopic rule applied to post-access-only records, take $\varepsilon_n=0$ in the current-revenue calculation \eqref{eq:revenue}. Full retention and neutral entry again yield
\begin{equation}
 c_n^*\geq R_n^*\geq R_0^*>0,
 \qquad D(b_n^*)\geq2R_0^*>0.
 \label{eq:baselinepricing}
\end{equation}
Thus the provider's myopic pricing supplies the uniform floor needed by Theorem~\ref{thm:nonlearning}. Merely assuming that each fee is positive would not suffice, because a positive sequence could converge to zero.

Under a fixed nonprohibitive fee, the cutoff is nondecreasing as the record improves, yet remains bounded by \eqref{eq:cap}. It converges to some $b_\infty>0$, while $L_n$ converges to a finite $L_\infty$ satisfying $|L_\infty|\geq b_\infty$. The probability that the next member contradicts the current recommendation tends to zero. This is an asymptotic statement: the argument does not require every history to enter an exact cascade at a finite date. Appendix~\ref{app:nonlearning} supplies the details.

The contrast with Theorem~\ref{thm:learning} can now be stated precisely. Positive prices bound the information of an individual buyer in both environments. This bound prevents complete learning from post-access actions, but does not prevent repeated independent preliminary directions from identifying the state. What matters is how the selected information is transformed before it is recorded.

\subsection{Occasional observations of nonbuyers}\label{sec:outsiders}

Consider a second change to the recording rule. Keep a fixed fee $0<c<V_0(0)$ and retain buyers' post-access actions. At chronological arrival $i$, an independent device observes a nonbuyer's actual private-signal action with probability $\rho_i\in[0,1]$. The record identifies the observation as a nonbuyer action and preserves it for later paying readers. The sampling schedule is deterministic and has no direct payoff effect. For this extension, let $Z_{i-1}$ be the public sequence of previous arrival labels: member, recorded nonbuyer, or unrecorded nonbuyer. Actions remain hidden until payment. At arrival $i$, access is valued conditional on $Z_{i-1}$, giving a cutoff $b_i(Z_{i-1})$. These labels preserve symmetry, but can change the distribution of archive strength and hence demand.\footnote{This is a technology for observing an action, not an assumption that an outsider sends a truthful costless report. The public labels identify the selection rule for each observation without revealing its direction. In this extension the cutoff can depend on the label history; no monotonicity in that history is required.}

For this comparison assume unbounded private beliefs. A price bounded away from zero leaves a positive tail of private signals outside at every history. Among nonbuyers at cutoff $b$, the private action is correct with probability
\begin{equation}
 q^{\rm out}(b)=\E[\sig(T)\mid T>b]\geq q_{\rm all}>1/2.
 \label{eq:outsideaccuracy}
\end{equation}
The excluded population is, in this sense, a particularly informative source of independent decisions.

\Needspace{8\baselineskip}
\begin{proposition}\label{prop:outsiders}
Under the sampling rule just described, full retention and unbounded private beliefs,
\begin{equation}
 \sum_i\rho_i=\infty
 \quad\Longrightarrow\quad
 \Pp(\theta=+1\mid\text{full record})\longrightarrow\ind\{\theta=+1\}
 \quad\text{almost surely}.
 \label{eq:outsiderlearning}
\end{equation}
If $\sum_i\rho_i<\infty$, only finitely many nonbuyer actions are recorded almost surely and complete learning fails.
\end{proposition}

Any constant positive observation probability is sufficient, as is a vanishing schedule such as $\rho_i=1/i$. The conclusion does not imply rapid learning when the sampling probability is small. It identifies a boundary: excluding nonbuyers entirely has a different asymptotic implication from observing a persistent, possibly sparse stream of their independent actions. Unbounded beliefs ensure a permanent source of excluded signals. With bounded beliefs, the outside population could disappear as information improves.

Advance access changes selection in another way. Selling eligibility before the private signal arrives keeps strong as well as weak signals in the contributing population. For a fixed experiment, the standard advance-selling comparison is between the expected information value $\E[V(X)]$ and spot revenue $\max_c c\Pp(V(X)\geq c)$. The former is larger when the positive information value is continuously distributed and nondegenerate; Appendix~\ref{app:advance} records the elementary argument. This is a useful contract-timing benchmark, but it addresses who becomes eligible rather than what is reported after exposure. In particular, admitting everyone does not by itself preserve independent actions once bounded signals are dominated by the history.

\section{Conclusion}\label{sec:conclusion}

An action can be valuable information for later individuals, but its information content depends on why the person acted and what was observed beforehand. When access to a history also determines eligibility to contribute, a positive price selects people with relatively weak private information. In the one-founder case, every buyer copies the founding recommendation. With a richer history, some buyers can contradict the recommendation and the record can become more accurate. A price bounded away from zero nevertheless prevents complete learning when only post-access actions are retained.

The same selected population can supply an informative record through consequential preliminary choices. A positive probability of implementing the choice made before access gives each buyer an incentive to follow private information at that stage. Under the fixed effective fee and pooled myopic pricing policies studied here, enough informative choices continue to enter for the archive to learn the state, including with bounded private beliefs. This changes the information supplied by customers without requiring them to purchase access before receiving their signals.

The reporting rule also changes the value of the service. Buyers are less willing to pay when they may have to implement a preliminary choice, and that possibility sacrifices current accuracy. The worked example shows how accumulating independent evidence can eventually improve both revenue and decisions despite these initial losses. The result provides a reason to distinguish judgments made before exposure from decisions made afterward, while keeping the costs and commitment requirements of that distinction explicit.

\clearpage
\appendix
\section{Proofs and additional calculations}\label{app:proofs}

\subsection{The value of observation and the founding benchmark}\label{app:value}

\begin{proof}[Proof of Lemma~\ref{lem:access}]
Fix $x\geq0$. Without access the individual chooses $+1$. Observing log odds $l$ changes the optimal action only when $l<-x$. Conditional on $X=x$, the joint probability measures of the positive and negative states and the archive are $e^x\nu_{+,n}(dl)/(1+e^x)$ and $\nu_{-,n}(dl)/(1+e^x)$. Switching to the negative action on $l<-x$ improves accuracy by their difference. Since $d\nu_{+,n}/d\nu_{-,n}(l)=e^l$, this gives \eqref{eq:generalvalue}.

Reflection gives $V_n(-x)=V_n(x)$. For $x\geq0$ the integrand can be written as $(1-e^{x+l})_+/(1+e^x)$, which is continuous and nonincreasing in $x$ and bounded by one. Dominated convergence establishes continuity. If the value is positive at a larger $x$, a set of realizations of positive measure contributes there and has strictly larger contributions at every smaller $x$. Thus the value is strictly decreasing while positive. Perfect information improves the private decision's accuracy by exactly $1-\sig(|x|)$; no experiment can do better, proving \eqref{eq:perfectbound}.

At a positive fee below $V_n(0)$, these properties give a unique positive cutoff, unless the entire bounded support is admitted. Boundary indifferences have probability zero. From \eqref{eq:signals},
\[
 \int_{-b}^b f_+(x)\,dx=\int_{-b}^b f_-(x)\,dx=D(b).
\]
The incoming signal is independent of the actual archive conditional on the state. Hence the probability of entry is $D(b)$ in either state, even after conditioning on that archive. Nonentry and waiting times have the same neutrality property. Starting from the symmetric initial experiment, the symmetric cutoff and optimal action rule preserve symmetry at subsequent sizes. This justifies the value calculation recursively.
\end{proof}

\begin{proof}[Proof of Proposition~\ref{prop:founder}]
Equation~\eqref{eq:foundervalue} follows by comparing the strength $|x|$ with $h$. If $|x|<h$, either realization of $Y$ determines the optimal action. Because $Y$ has accuracy $r$ in each state and is conditionally independent of $X$, following $Y$ also has accuracy $r$ conditional on $X=x$. If $|x|\geq h$, the individual's action is the same with or without access. The entry cutoff is therefore \eqref{eq:foundercutoff}.

Every buyer's signal has strength below $h$, so every buyer chooses $Y$. Entry has likelihood ratio one by Lemma~\ref{lem:access}; conditional on $Y$, the new action is deterministic and also has likelihood ratio one. Induction gives \eqref{eq:freeze}. Entry probability $D(b(c))$ is positive and the same at every arrival. Conditional independence yields infinitely many buyers almost surely. All their actions are incorrect exactly when $Y\ne\theta$, an event of probability $1-r$.
\end{proof}

The confidence-disclosure calculation in Section~\ref{sec:confidence} is the same conditional argument. Symmetry makes $H=|L_0|$ uninformative about the state. Given $H=h$, the sign of $L_0$ is a binary experiment with state-conditional accuracy $\sig(h)$. Residual distinctions between initial histories with the same log odds carry no further information about the binary state. Applying the preceding proof conditional on $h$ gives \eqref{eq:disclosedvalue} and exact freezing at any positive fee with positive demand. At $h=0$, the conditional experiment has no value.

\subsection{Preliminary choices and the joint likelihood}\label{app:elicitation}

\begin{proof}[Proof of Proposition~\ref{prop:elicitation}]
The preliminary choice has positive payoff weight $\varepsilon_n$ and cannot affect the fee, the archive snapshot, the revised decision problem or any later payoff of the same individual. Its optimal direction is therefore the sign of the private posterior; the gain over the opposite report is $\varepsilon_n\tanh(|x|/2)$, strictly positive away from $x=0$. After access, the revised decision has positive weight $1-\varepsilon_n$ and posterior log odds $x+L_n$. Ties have probability zero under the continuous signal distribution. This proves \eqref{eq:twodecisions}.

Taking expectations before access gives \eqref{eq:payoff}. Subtracting the fee and comparing with the nonbuyer payoff gives $V_n(x)\geq c_n/(1-\varepsilon_n)$. Lemma~\ref{lem:access} supplies the cutoff and neutrality conclusions. The preliminary choice is made after payment, so neither access nor its cost is contingent on its direction. No inference about the reporter's honesty or competence is used in this incentive argument.
\end{proof}

Conditional on entry and an actual history with log odds $l$, the state-conditional signal density is $f_\theta(x)\ind\{|x|\leq b\}/D(b)$. Therefore the probability of a pair $(u,v)$ is
\begin{equation}
 \Pp_\theta(B=u,A=v\mid\HH_n,\text{entry})
 =\frac{\int_{C_{u,v}(l,b)}f_\theta(x)\,dx}{D(b)}.
 \label{eq:jointprob}
\end{equation}
The entry denominators cancel in the likelihood ratio, yielding \eqref{eq:pairupdate}. The preliminary marginal gives accuracy $q_b$ as in \eqref{eq:selectedaccuracy}. If $|l|\geq b$, the revised choice is deterministic and the pair's information comes entirely from that marginal. If $|l|<b$, the revised choice can further restrict the signal interval conditional on the preliminary sign. Adding two marginal log-likelihood ratios would generally count dependent evidence incorrectly.

\subsection{Pricing bounds and learning}\label{app:learning}

\begin{proof}[Proof of Theorem~\ref{thm:learning}]
We first establish the participation bounds. At size $n$, entry has equal probability in each state and is independent of the actual archive conditional on the state. Selecting the next member therefore leaves the state-conditional law of $\HH_n$ unchanged. The next record contains $\HH_n$ as a subrecord. Discarding the new pair reproduces the preceding experiment. Combining either experiment with an independent incoming private signal preserves this ordering, so \eqref{eq:refinement} holds.

Under the fixed effective fee, $b_n\geq b_0>0$, and we can take $b_{\min}=b_0$. Under pooled myopic pricing, positive information value near $x=0$ gives $R_0^*>0$. Revenue attains a maximum on $[0,1/2]$: demand is continuous at positive values on the range of $V_n$ because the signal distribution is continuous and $V_n$ strictly decreases where positive; all-buy boundaries add no atom; and revenue tends to zero as the price tends to zero. The provider can obtain at least initial optimal revenue by keeping the initial effective fee, so $R_n^*\geq R_0^*$. Since $R_n^*=p_n^*D(b_n^*)$, with $D\leq1$ and $p_n^*\leq1/2$, inequalities \eqref{eq:positivebounds} follow. Choose $b_{\min}=D^{-1}(2R_0^*)>0$. The inverse exists in the signal support; in particular $R_0^*<1/2$ because the information value is strictly below $1/2$ almost surely for a finite, imperfect seed and a continuous private signal.

Thus each next member arrives after a finite waiting time almost surely, with entry probability bounded below by $D(b_{\min})>0$. A countable intersection gives infinitely many completed member records. Under a fixed effective fee, and under the deterministic pooled-price tie rule, the cutoff sequence is a deterministic function of member size. Conditional on the state, the private signal selected at each new membership is independent of preceding selected signals and has the corresponding truncated density. Its preliminary sign is correct with probability $q_{b_n}\geq q_{\min}>1/2$.

Consider the estimator that ignores the seed and revised choices and uses the majority of $n$ preliminary signs. Conditional on either state, let $Z_j$ indicate that the $j$th preliminary sign is correct. The $Z_j$ are independent, bounded between zero and one, and have means at least $q_{\min}$. With $d=q_{\min}-1/2$, the standard exponential bound for independent bounded variables gives
\[
 \Pp_\theta\left(\sum_{j=1}^n Z_j\leq n/2\right)
 \leq\inf_{t>0}\exp(-tnd+nt^2/8)
 =\exp(-2nd^2).
\]
This treats a tied majority as an error, so it bounds any tie rule. The Bayes estimator based on the entire archive has no greater ex ante error. Equation~\eqref{eq:errorbound} follows.

Let $\pi_n=\Pp(\theta=+1\mid\HH_n)$. Full retention makes $\pi_n$ a bounded martingale, so it converges almost surely. The error bound implies $\E[\min(\pi_n,1-\pi_n)]\to0$. Bounded convergence gives $\pi_\infty\in\{0,1\}$ almost surely. Moreover, $\pi_\infty=\E[\ind\{\theta=+1\}\mid\HH_\infty]$. Conditional-expectation calibration then implies $\pi_\infty=\ind\{\theta=+1\}$ almost surely, proving \eqref{eq:learning}.

Finally, an entering member could ignore the private signal and use the archive's Bayes action. Conditioning on entry does not change the archive's ex ante distribution, because entry is independent of its realization conditional on the state and equally likely in both states. The optimal revised-choice error is therefore no larger than the archive-only Bayes error and tends to zero.
\end{proof}

For completeness, the lottery-schedule invariance used in Section~\ref{sec:cost} follows by induction. At a fixed archive experiment, the choices in \eqref{eq:twodecisions} and the entry condition at a given effective fee do not depend on the strictly positive magnitude of $\varepsilon_n$. The pooled effective-price maximizer is independent of it as well. Starting from the same initial experiment and using the same effective-price rule, \eqref{eq:jointprob} then generates the same distribution of the next record at every size. The lotteries are not revealed during record collection. Changing their probabilities affects implemented choices and actual revenue, but not the statistical law of the records.

\subsection{Implementation cost and the worked example}\label{app:example}

\begin{proof}[Proof of Proposition~\ref{prop:cost}]
The lottery is independent of both decisions, so \eqref{eq:implementederror} is exact. Since $q_{b_n}\leq q_{\rm all}<1$ and $e_{A,n}\to0$ by Theorem~\ref{thm:learning}, a fixed $\varepsilon>0$ gives \eqref{eq:errorfloor}. If instead $\varepsilon_n\to0$, both terms in \eqref{eq:implementederror} vanish. Conditional on $X=x$, revised accuracy exceeds preliminary accuracy in expectation by $V_n(x)$. Averaging over the entry set gives \eqref{eq:cost}, whose upper bound follows from \eqref{eq:perfectbound}. Maximizing $(1-\varepsilon_n)R_n(p)$ at a fixed experiment gives the revenue comparison.
\end{proof}

We give the example's accuracy calculation explicitly. Under \eqref{eq:exampledensity}, elementary integration gives \eqref{eq:exampleclosed} and $q_{\rm all}=\sig(1/2)$. At $r=3/4$ and effective fee $p=0.10$, the initial cutoff satisfies $\sig(b_0)=r-p=0.65$, giving $b_0=\log(13/7)$. With only post-access decisions retained, every buyer has accuracy $r$. Define
\begin{equation}
 D_0=D(b_0),\qquad G_0=D_0(r-q_{b_0}).
 \label{eq:initialgain}
\end{equation}
Relative to private-only decisions throughout the population, the initial archive raises expected accuracy by $G_0$. Thus the stationary baseline outcomes per arriving individual are
\begin{equation}
 W^{\rm base}=q_{\rm all}+G_0=0.727006,
 \qquad \Pi^{\rm base}=pD_0=0.060351.
 \label{eq:baseoutcomes}
\end{equation}
The preliminary-choice rule begins with the same entry set, so its initial outcomes are
\begin{equation}
 W_0=q_{\rm all}+(1-\varepsilon)G_0=0.721779,
 \qquad \Pi_0=(1-\varepsilon)pD_0=0.057333.
 \label{eq:startoutcomes}
\end{equation}

Under any fixed effective fee satisfying Theorem~\ref{thm:learning}, complete learning implies, for every fixed $x$,
\begin{equation}
 V_n(x)\longrightarrow1-\sig(|x|),\qquad
 b_n\longrightarrow b_\infty=\min\left\{\bar x,\log\frac{1-p}{p}\right\}.
 \label{eq:limitcutoff}
\end{equation}
Indeed, the archive-only Bayes error vanishes under each state. Conditional on fixed $X=x$, using that Bayes action therefore approaches accuracy one; optimal accuracy is at least this large and at most one. This proves the value limit. Monotonicity and continuity of the limiting value give the cutoff limit, with equality types irrelevant under the continuous signal distribution.

In the example, the support-edge limiting value $1-\sig(1)$ strictly exceeds $p$. Consequently $b_n=1$ for all sufficiently large $n$ and the entire population buys. Revised error then vanishes, whereas preliminary accuracy remains $\sig(1/2)$. This proves the limiting accuracy \eqref{eq:examplelimit} and revenue $(1-\varepsilon)p=0.095$.

The finite rows of Table~\ref{tab:comparison} start from the two founder realizations with probabilities $(r,1-r)$ in the positive state and reversed probabilities in the negative state. At each size, the program calculates $V_n$ from the full distribution, solves the effective entry condition, and applies \eqref{eq:jointprob} to every history. For an interval $[u,v]\subset[-1,1]$, its probabilities are
\begin{equation}
 \int_u^v f_+(x)\,dx=\frac{e^{v/2}-e^{u/2}}{2\sinh(1/2)},\qquad
 \int_u^v f_-(x)\,dx=\frac{e^{-u/2}-e^{-v/2}}{2\sinh(1/2)}.
 \label{eq:intervalprob}
\end{equation}
Thus no numerical integration or simulation is required. The recursion retains all 110,422 histories at twelve member records, without pruning. At that size, archive-only accuracy is $0.821113$, revised accuracy is $0.829940$, and implemented accuracy is $0.819566$. Distinguishing these objects explains why reaching full participation does not mean that decisions are already close to their limiting accuracy.

The same policy comparison can be expressed without the particular density. For a single-founder baseline with a nonempty proper entry set, put
\begin{equation}
 D_\infty=D(b_\infty),\qquad
 G_\infty=\E[(1-\sig(T))\ind\{T\leq b_\infty\}].
 \label{eq:finalgain}
\end{equation}
The preliminary-choice rule's limiting population accuracy is $q_{\rm all}+(1-\varepsilon)G_\infty$, and its revenue is $(1-\varepsilon)pD_\infty$. When $D_\infty>D_0$ and $G_\infty>G_0$, both exceed their stationary baseline counterparts whenever
\begin{equation}
 0<\varepsilon<\min\left\{1-\frac{D_0}{D_\infty},\,
                           1-\frac{G_0}{G_\infty}\right\}.
 \label{eq:epsiloncomparison}
\end{equation}
The two bounds in the example are $0.396491$ and $0.723084$. Hence $\varepsilon=0.05$ satisfies both with slack. These are long-run comparisons under fixed policies; no continuation-value optimization is used.

\subsection{The bound on post-access learning}\label{app:nonlearning}

\begin{proof}[Proof of Theorem~\ref{thm:nonlearning}]
Lemma~\ref{lem:access} implies \eqref{eq:cap}. Given current log odds $l$ and the effective support $[-b_n,b_n]$ of buyers, the probability of the positive action is \eqref{eq:actionprob}. On $|l|<b_n$, the truncated signal distributions have a strict monotone likelihood ratio, giving $0<Q_-<Q_+<1$. Both observed actions therefore have a nontrivial Bayesian likelihood increment. For either action event $C\subset[-b_n,b_n]$,
\begin{equation}
 \frac{\int_C f_+(x)\,dx}{\int_C f_-(x)\,dx}
 =\frac{\int_C e^x f_-(x)\,dx}{\int_C f_-(x)\,dx}
 \in[e^{-b_n},e^{b_n}].
 \label{eq:incrementbound}
\end{equation}
Thus every increment has absolute value at most $b_n\leq\bar b$.

If $|L_n|\geq\bar b$, all possible buyers follow its direction. Entry is neutral and the new action is deterministic, so $L_{n+1}=L_n$. If $|L_n|<\bar b$, the increment bound gives $|L_{n+1}|<2\bar b$. Induction proves \eqref{eq:globalcap}. Nonentry supplies no additional likelihood increment. If only finitely many members enter, the remaining archive is constant and the same conclusion holds.

The posterior based on the full retained record is a bounded martingale and hence converges. Since \eqref{eq:globalcap} is finite along almost every path, its limit is strictly between zero and one. If $|L_0|\leq K_0$ surely, applying $\sig$ gives $\sig(-K)\leq\Pp(\theta=+1\mid\HH_n)\leq\sig(K)$, so both the conditional and ex ante archive-only Bayes errors are at least $\sig(-K)$. This proves the stated uniform lower bound.
\end{proof}

The pooled-price floor \eqref{eq:baselinepricing} follows by the same refinement and maximization argument as in Appendix~\ref{app:learning}, now with only post-access actions and $\varepsilon_n=0$. It concerns the actual distribution generated by that recording policy. We do not identify it with the distribution generated by preliminary records.

We finish the fixed-fee limiting argument in the main text. At $0<c<V_0(0)$, refinement makes $b_n$ nondecreasing, with $b_0\leq b_n\leq\bar b$. Hence $b_n\to b_\infty>0$. The lower bound on entry probability ensures infinitely many members. Theorem~\ref{thm:nonlearning} gives $L_n\to L_\infty\in\mathbb R$. Suppose $|L_\infty|<b_\infty$. Eventually $(L_n,b_n)$ lies in a compact subset of the region where both actions have positive probability and $Q_+-Q_->0$. Continuity makes the absolute log-likelihood increment from either action bounded away from zero there. Infinitely many such updates contradict convergence of $L_n$. Therefore $|L_\infty|\geq b_\infty$.

If the inequality is strict, every member eventually copies. At equality, the interval of admitted signals capable of contradicting the current recommendation shrinks to an endpoint. Its conditional probability tends to zero by continuity of the signal distribution and $D(b_n)\geq D(b_0)>0$. The equality case does not require copying to begin at a finite date.

\subsection{Observing a sparse set of nonbuyers}\label{app:outsiders}

\begin{proof}[Proof of Proposition~\ref{prop:outsiders}]
Under either state, the magnitude $T$ has the same distribution. Conditional on $T=t$, the private action is correct with probability $\sig(t)$. Truncating below at $b$ therefore gives \eqref{eq:outsideaccuracy}, and $q^{\rm out}(b)$ is nondecreasing in $b$. To verify the label conditioning, fix the full past and its public labels $Z_{i-1}$. Given the common cutoff $b_i=b_i(Z_{i-1})$, the next labels have state-conditional probabilities $D(b_i)$, $[1-D(b_i)]\rho_i$ and $[1-D(b_i)](1-\rho_i)$. These probabilities are the same in both states, including conditional on the actual past archive. Reflection preserves the conditional archive experiment at every public label history, so Lemma~\ref{lem:access} applies recursively. Labels can affect the distribution of the information product without themselves supplying a directional likelihood increment.

At fixed fee $c>0$, Lemma~\ref{lem:access} bounds all buyer signals by $\bar b(c)=\log((1-c)/c)$. Because private beliefs are unbounded, $\delta=\Pp(T>\bar b(c))>0$. Signals in this tail never buy, whatever the current archive. Conditional on either state, the events that arrival $i$ has such a signal and its independent sampling coin succeeds are independent across arrivals, with probabilities $\delta\rho_i$. If $\sum_i\rho_i=\infty$, the independent-events Borel--Cantelli argument gives infinitely many recorded outsider actions almost surely. This establishes a permanent supply even though the full nonbuyer population is endogenously selected.

Use the filtration containing the entire record, the applicable access cutoffs and nondirectional public labels. An outsider action recorded at cutoff $b_i$ changes log odds by
\[
 \pm\log\frac{q^{\rm out}(b_i)}{1-q^{\rm out}(b_i)}.
\]
Its absolute value is at least $\log(q_{\rm all}/(1-q_{\rm all}))>0$. All cutoffs and sampling labels are known and state-neutral before their associated directional observation, so these are the appropriate conditional likelihood increments. The full posterior is a bounded martingale. If it converged to an interior value, finite log odds would converge and successive increments would vanish, contradicting infinitely many outsider increments bounded away from zero. Its limit must therefore lie in $\{0,1\}$; conditional-expectation calibration, as in the proof of Theorem~\ref{thm:learning}, identifies the endpoint with the true state.

If $\sum_i\rho_i<\infty$, only finitely many sampling coins succeed almost surely, so only finitely many outsiders are recorded. Each gives a finite likelihood ratio: the cutoff is bounded and both signal tails have positive probability. After the last outsider observation, the current log odds are finite. The subsequent path consists of buyer actions with the same signal cap, together with nondirectional labels and nonentry. The pathwise argument of Theorem~\ref{thm:nonlearning} bounds the remaining log odds by the larger of their current magnitude and $2\bar b(c)$. Thus the record cannot converge to full revelation.
\end{proof}

\subsection{The advance-access comparison}\label{app:advance}

Fix a symmetric archive experiment and let $V=V(X)$ denote its value after the private signal. Before the signal, otherwise identical individuals value access at $\E[V]$, so any advance fee strictly below this amount is accepted. Suppose the positive part of $V$ is continuously distributed and nondegenerate, as under the maintained densities and a nontrivial imperfect archive. For a positive spot fee $c$ with demand,
\begin{equation}
 \E[V]-c\Pp(V\geq c)
 =\E[(V-c)\ind\{V\geq c\}]+\E[V\ind\{0<V<c\}]>0.
 \label{eq:advance}
\end{equation}
The spot revenue function attains a positive maximum. The strict inequality at its maximizer leaves room for an advance fee below $\E[V]$ that earns more. This is a comparison at a fixed experiment and does not require a dynamic contract. It removes selection on the realized signal when eligibility is sold. By itself, it makes no guarantee that later post-exposure decisions continue to convey independent information when signals are bounded.


\begin{thebibliography}{99}
\bibitem[Acemoglu et~al.(2022)Acemoglu, Makhdoumi, Malekian, and Ozdaglar]{acemoglu2022}
Acemoglu, D., Makhdoumi, A., Malekian, A., \& Ozdaglar, A. (2022). Learning from reviews: The selection effect and the speed of learning. \emph{Econometrica}, \emph{90}(6), 2857--2899. \url{https://doi.org/10.3982/ECTA15847}.

\bibitem[Arieli et~al.(2022)Arieli, Koren, and Smorodinsky]{arieli2022}
Arieli, I., Koren, M., \& Smorodinsky, R. (2022). The implications of pricing on social learning. \emph{Theoretical Economics}, \emph{17}(4), 1761--1802. \url{https://doi.org/10.3982/TE3842}.

\bibitem[Azrieli et~al.(2018)Azrieli, Chambers, and Healy]{azrieli2018}
Azrieli, Y., Chambers, C. P., \& Healy, P. J. (2018). Incentives in experiments: A theoretical analysis. \emph{Journal of Political Economy}, \emph{126}(4), 1472--1503. \url{https://doi.org/10.1086/698136}.

\bibitem[Bikhchandani et~al.(1992)Bikhchandani, Hirshleifer, and Welch]{bhw1992}
Bikhchandani, S., Hirshleifer, D., \& Welch, I. (1992). A theory of fads, fashion, custom, and cultural change as informational cascades. \emph{Journal of Political Economy}, \emph{100}(5), 992--1026. \url{https://doi.org/10.1086/261849}.

\bibitem[Bikhchandani et~al.(2024)Bikhchandani, Hirshleifer, Tamuz, and Welch]{bikhchandani2024}
Bikhchandani, S., Hirshleifer, D., Tamuz, O., \& Welch, I. (2024). Information cascades and social learning. \emph{Journal of Economic Literature}, \emph{62}(3), 1040--1093. \url{https://doi.org/10.1257/jel.20241472}.

\bibitem[Frey and van de Rijt(2021)]{frey2021}
Frey, V., \& van de Rijt, A. (2021). Social influence undermines the wisdom of the crowd in sequential decision making. \emph{Management Science}, \emph{67}(7), 4273--4286. \url{https://doi.org/10.1287/mnsc.2020.3713}.

\bibitem[Guarino et~al.(2011)Guarino, Harmgart, and Huck]{guarino2011}
Guarino, A., Harmgart, H., \& Huck, S. (2011). Aggregate information cascades. \emph{Games and Economic Behavior}, \emph{73}(1), 167--185. \url{https://doi.org/10.1016/j.geb.2011.01.003}.

\bibitem[Herrera and H\"orner(2013)]{herrera2013}
Herrera, H., \& H\"orner, J. (2013). Biased social learning. \emph{Games and Economic Behavior}, \emph{80}, 131--146. \url{https://doi.org/10.1016/j.geb.2012.12.006}.

\bibitem[Ifrach et~al.(2019)Ifrach, Maglaras, Scarsini, and Zseleva]{ifrach2019}
Ifrach, B., Maglaras, C., Scarsini, M., \& Zseleva, A. (2019). Bayesian social learning from consumer reviews. \emph{Operations Research}, \emph{67}(5), 1209--1221. \url{https://doi.org/10.1287/opre.2019.1861}.

\bibitem[Kultti and Miettinen(2007)]{kultti2007}
Kultti, K. K., \& Miettinen, P. A. (2007). Herding with costly observation. \emph{The B.E. Journal of Theoretical Economics}, \emph{7}(1), Article 28. \url{https://doi.org/10.2202/1935-1704.1320}.

\bibitem[Markovich and Yehezkel(2024)]{markovich2024}
Markovich, S., \& Yehezkel, Y. (2024). ``For the public benefit'': Data policy in platform markets. \emph{Journal of Economics \& Management Strategy}, \emph{33}(3), 652--685. \url{https://doi.org/10.1111/jems.12588}.

\bibitem[Nocke et~al.(2011)Nocke, Peitz, and Rosar]{nocke2011}
Nocke, V., Peitz, M., \& Rosar, F. (2011). Advance-purchase discounts as a price discrimination device. \emph{Journal of Economic Theory}, \emph{146}(1), 141--162. \url{https://doi.org/10.1016/j.jet.2010.07.008}.

\bibitem[Parakhonyak and Vikander(2023)]{parakhonyak2023}
Parakhonyak, A., \& Vikander, N. (2023). Information design through scarcity and social learning. \emph{Journal of Economic Theory}, \emph{207}, Article 105586. \url{https://doi.org/10.1016/j.jet.2022.105586}.

\bibitem[Peng et~al.(2025)Peng, Rao, Sun, and Xiao]{peng2025}
Peng, D., Rao, Y., Sun, X., \& Xiao, E. (2025). Optional disclosure and observational learning. \emph{Journal of Economic Behavior \& Organization}, \emph{229}, Article 106817. \url{https://doi.org/10.1016/j.jebo.2024.106817}.

\bibitem[Smirnov and Starkov(2025)]{smirnov2025}
Smirnov, A., \& Starkov, E. (2025). Designing social learning. \emph{European Economic Review}, \emph{178}, Article 105113. \url{https://doi.org/10.1016/j.euroecorev.2025.105113}.

\bibitem[Smith and S\o rensen(2000)]{smith2000}
Smith, L., \& S\o rensen, P. N. (2000). Pathological outcomes of observational learning. \emph{Econometrica}, \emph{68}(2), 371--398. \url{https://doi.org/10.1111/1468-0262.00113}.

\bibitem[Song(2016)]{song2016}
Song, Y. (2016). Social learning with endogenous observation. \emph{Journal of Economic Theory}, \emph{166}, 324--333. \url{https://doi.org/10.1016/j.jet.2016.09.005}.

\bibitem[Xie and Shugan(2001)]{xie2001}
Xie, J., \& Shugan, S. M. (2001). Electronic tickets, smart cards, and online prepayments: When and how to advance sell. \emph{Marketing Science}, \emph{20}(3), 219--243. \url{https://doi.org/10.1287/mksc.20.3.219.9765}.
\end{thebibliography}
\end{document}